\documentclass[conference]{IEEEtran}
\usepackage{amsfonts}
\usepackage{graphicx}
\usepackage{color}
\usepackage{amsmath}
\usepackage{authblk}
\usepackage{theorem}
\usepackage{color}
\newtheorem{theorem}{Theorem}
\newtheorem{lemma}{Lemma}

\usepackage{multicol}
\begin{document}
\title{Rate Selection for Cooperative HARQ-CC Systems over Time-Correlated Nakagami-m Fading Channels}
\author{Zheng Shi\thanks{shizheng0124@gmail.com}}
\author{Shaodan Ma}
\author{Kam-Weng Tam}
\affil{Department of Electrical and Computer Engineering, University of Macau, Macau}
\maketitle
\begin{abstract}
This paper addresses the problem of rate selection for the cooperative hybrid automatic repeat request with chase combination (HARQ-CC) system, where time correlated Nakagami-m fading channels are considered. To deal with this problem, the closed-form cumulative distribution function (CDF) for the combine SNRs through maximal ratio combining (MRC) is first derived as a generalized Fox's $\bar H$ function. By using this result, outage probability and delay-limited throughput (DLT) are derived in closed forms, which then enables the rate selection for maximum DLT. These analytical results are validated via Monte Carlo simulations. The impacts of time correlation and channel fading-order parameter $m$ upon outage probability, DLT and the optimal rate are investigated thoroughly. It is found that the system can achieve more diversity gain from less correlated channels, and the outage probability of cooperative HARQ-CC system decreases with the increase of $m$, and etc. Furthermore, the optimal rate increases with the number of retransmissions, while it decreases with the increase of the channel time correlation.
\end{abstract}
\IEEEpeerreviewmaketitle
\section{Introduction}
Hybrid automatic repeat request (HARQ) is a reliable and effective transmission technique to overcome the channel impairments without the knowledge of channel state information (CSI). In HARQ scheme, the transmission information is first encoded via forward error-correcting encoder. The encoded packet will be transmitted in a number of HARQ rounds until the receiver successfully decodes the packet. On the basis of whether the previously failed packets are used for decoding or not, HARQ scheme can be classified into two types, i.e. the Type I and the Type II HARQ. For the Type I, the previously failed packets will be discarded. On the contrary, for the Type II, the previously failed packets will be stored for reutilization. By using different decoding techniques to deal with all the received packets, the Type II can be further divided into two schemes, i.e. HARQ with chase combining (HARQ-CC) and HARQ with incremental redundancy (HARQ-IR). For HARQ-CC, the same packet will be transmitted in each HARQ round, and the technique of maximal rate combination (MRC) is adopted to combine all the received packets. However, new parity bits will be sent in each HARQ round for HARQ-IR, and code combining method is used for decoding the received packets. It is obvious that the Type II is superior to Type I because of exploiting the previously failed packets. Meanwhile, HARQ-IR generally performs better than HARQ-CC in terms of data transmission rate because of utilizing complex code combining approach. However, the HARQ-CC scheme has a low complexity, and a low hardware requirement, especially the buffer size. Thus our focus is put onto HARQ-CC in this paper.

Performance of HARQ-CC has been frequently discussed in the literature. Generally, the performance evaluation of HARQ-CC depends on the distribution of the combining signal-to-noise ratio (SNR) through MRC. In \cite{su2011optimal}, the optimal power assignment is presented for quasi-static Rayleigh fading channels. Because of the channels are constant in each HARQ round, the combine SNR follows exponential distribution. In \cite{chaitanya2013optimal}, optimal power allocation is studied for i.i.d. Rayleigh fading channels. To simplify the analysis, a truncated probability density function (PDF) of the combine SNR was derived. The truncated PDF become more accurate for a higher transmitted SNR. Both the HARQ-IR and the HARQ-CC schemes are studied in cooperative system over i.i.d. Rayleigh fading channels in \cite{chelli2013performance}, where the exact PDF of the combine SNR is derived for evaluating HARQ-CC.

In the practical scenarios, users with low-to-medium mobility typically suffer from time-correlated fading \cite{kim2011optimal}. An optimal rate adaption scheme is proposed for time-correlated Rayleigh fading channels, where the optimal rate is obtained through maximizing the delay-limited throughput (DLT) by using the Gaussian approximation for the combine SNR. In \cite{jin2011optimal}, another optimal rate selection scheme is presented for time-correlated Nakagami-m fading channels via minimizing the allocated resource. In that paper, the combine SNR is approximated as a Gamma random variable. An adaptive power allocation is proposed for time-correlated Rayleigh fading channels in \cite{chaianya2014adaptive} by using the same approximation for the combine SNR as \cite{chaitanya2013optimal}.

Literatures seldom put their interest upon the cooperative HARQ systems, especially for the time-correlated channels. Due to the involvement of cooperative communications and the presence of time-correlation, the rate selection becomes more challenging than the prior works. In this paper, we investigate the rate selection problem in the cooperative HARQ-CC system over time-correlated Nakagami-m fading channels. The PDF of the combine SNR is first derived in closed-form as a generalized Fox's $\bar H$ function. Based on the analytical result, outage probability and DLT are then derived, which then enables the rate selection by maximizing the DLT. The impacts of time-correlation, the maximum number of transmissions and channel fading-order parameter $m$ on outage probability, DLT and the optimal rate are finally investigated for cooperative HARQ-CC system over time-correlated Nakagami-m fading channels.

The remainder of the paper is organized as follows. In Section \ref{sys_mod}, the system model is presented. In Section \ref{sec:per}, the distribution of the instantaneous combine SNR via MRC can be derived in terms of the generalized Fox's $\bar H$ function. By applying these analytical expressions into performance analysis, the outage probability and DLT for HARQ-CC are derived in Section \ref{adap_sch}, and an optimal rate selection scheme is presented. The Section \ref{sec_con} draws some important conclusions.

\section{System Model}
\label{sys_mod}
\begin{figure}
  \centering
  \includegraphics[width=3.5in]{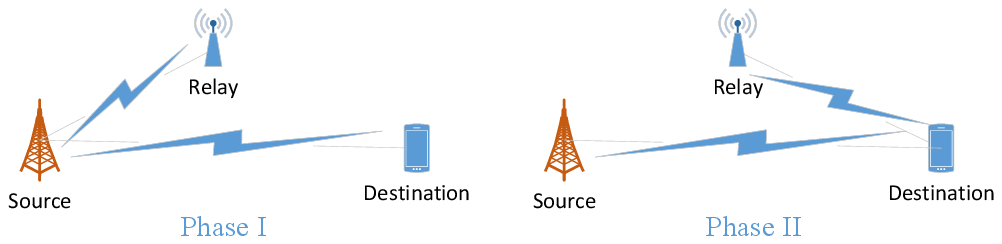}\\
  \caption{System model}\label{fig:sys_mod}
\end{figure}
In this paper, we consider a cooperative system of three nodes, as shown in Fig. \ref{fig:sys_mod}. Prior to the transmission of $b$ information bits to the destination, these bits are encoded and transmitted with a rate of $\mathcal R$. Denote the transmitted signal as $x$. The power of transmitted signal $x$ is assumed to be $P_s = {\rm{E}}(x^2)$. The same signal will be transmitted in each HARQ round. The transmission procedure consists of two phase by employing HARQ-CC into cooperative system. For the phase I, the source broadcasts the transmitted signal $x$ to the relay and the destination. The received signal corrupted by the block Nakagami-m fading channels at the destination and the relay in the $k$th HARQ round can be expressed as
\begin{equation}\label{eqn:rec_signal_k_round_SD}
  y_{SD}^k = h_{SD}^k x + n_{SD}^k
\end{equation}
\begin{equation}\label{eqn:rec_signal_k_SR}
  y_{SR}^k = h_{SR}^k x + n_{SR}^k
\end{equation}
where $h_{SD}^k$ and $h_{SR}^k$ represent the channel coefficients associated with the links of source-to-destination and source-to-relay, respectively; $n_{SD}^k$ and $n_{SR}^k$ are defined as additive white gaussian noises with mean zero and variance $\sigma^2$, i.e. $n_{SD}^k~(n_{SR}^k) \sim N(0,\sigma^2)$. Unlike the previous research upon HARQ-CC scheme \cite{su2011optimal,chaitanya2013optimal,chelli2013performance}, time-correlation of channels is considered in this paper. The time-correlated channels are modeled as a multivariate Nakagami-m distribution with exponential correlation. Specifically, the amplitude of channel gain $|h_{SD}^{k}|$ follows a Nakagami-m distribution, i.e., $|h_{SD}^{k}| \sim Nakagami\left( m,{\Omega _{SD}^k} \right)$. Hereby, the squared amplitude of channel gain $\vert h_{SD}^{k}\vert^{2}$ are Gamma-distributed with parameters $(m,\Omega_{SD}^{k}/{m})$, i.e. $\vert h_{SD}^{k}\vert^{2} \sim \mathcal{G} (m,\Omega_{SD}^{k}/{m})$. Here a widely used exponential time correlation model \cite{beaulieu2011novel} is adopted and the cross-correlation between $|h_{SD}^i|^2$ and $|h_{SD}^j|^2$ can be derived as
\begin{equation}\label{eq:4}
{\rho _{SD}^{i,j}} = \frac{{cov(|{h_{SD}^i}{|^2},|{h_{SD}^j}{|^2})}}{{\sqrt {var(|{h_{SD}^i}{|^2})var(|{h_{SD}^j}{|^2})} }} = {\lambda _{SD}^i}^2{\lambda _{SD}^j}^2
\end{equation}

The same time-correlation model applies to $|h_{SR}^k|$, that is, $|h_{SR}^{k}| \sim Nakagami\left( m,{\Omega _{SR}^k} \right)$, and the cross-correlation between $|h_{SR}^i|^2$ and $|h_{SR}^j|^2$  is denoted by ${\rho _{SR}^{i,j}} = {\lambda _{SR}^i}^2{\lambda _{SR}^j}^2$.

In order to decode the received signals at the destination in this phase, MRC is adopted to combine all the previously received signals. Assume that $M$ HARQ rounds are utilized, the combining signal is therefore formulated as
\begin{equation}\label{eqn:rec_MRC_foru}
r_{SD} = {{\bf{w}}^T}{{\bf{y}}_{SD}} = {{\bf{w}}^T}{\bf{h}}_{SD}x + {{\bf{w}}^T}{\bf{n}}_{SD}
\end{equation}
where ${\bf{y}}_{SD}=\{y_{SD}^1,y_{SD}^2,\cdots,y_{SD}^M\}^T$ defines the vector of received signals; ${\bf{w}} = \{w_1, w_2, \cdots, w_M\}^T$ defines the weights of MRC; ${\bf{h}}_{SD}=\{h_{SD}^1,h_{SD}^2,\cdots,h_{SD}^M\}^T$ and ${\bf{n}}_{SD}=\{n_{SD}^1,n_{SD}^2,\cdots,n_{SD}^M\}^T$ are defined as column vectors of channel coefficients and Gaussian noises in all HARQ rounds, respectively. Since the signal $x$ has average power $P_s$, the instantaneous output SNR for MRC after $M$ HARQ rounds is
\begin{equation}\label{eqn:snr_M}
{\gamma _{D,1}^M} = \frac{{{P_s}{{\left| {{{\bf{w}}^T}{{\bf{h}}_{SD}}} \right|}^{\rm{2}}}}}{{\mathbb{E}\left\{ {{{\left| {{{\bf{w}}^T}{{\bf{n}}_{SD}}} \right|}^2}} \right\}}} = \frac{{{P_s}{{\left| {{{\bf{w}}^T}{{\bf{h}}_{SD}}} \right|}^{\rm{2}}}}}{{{\sigma ^2}{{\left\| {\bf{w}} \right\|}^2}}}
\end{equation}

Since constants do not matter, one could always scale $\bf{w}$ such that ${{{\left\| {\bf{w}} \right\|}^2}}=1$. For MRC technique, the output SNR $\gamma _M$ is maximized. By using Cauchy-Schwarz inequality, the expression (\ref{eqn:snr_M}) has a maximum only if $\bf{w}$ is linearly proportional to ${\bf{h}}_{SD}$, i.e. ${w_k} = \frac{{h_{SD}^k}^*}{{\left\| {{{\bf{h}}_{SD}}} \right\|}}$. Therefore, $\gamma_{D,1}^M$ is rewritten as
\begin{equation}\label{eqn:snr_com_fina}
{\gamma _{D,1}^M} = \frac{{{P_s}{{\left\| {{{\bf{h}}_{SD}}} \right\|}^2}}}{{{\sigma ^2}}} = \frac{{{P_s}}}{{{\sigma ^2}}}\sum\nolimits_{k = 1}^M {{{\left| {h_{SD}^k} \right|}^2}}
\end{equation}

Similarly, the combine SNR at the relay after $M$ HARQ rounds is formulated as
\begin{equation}\label{eqn:snr_com_fin_sr}
{\gamma _{R}^M} = \frac{{{P_s}}}{{{\sigma ^2}}}\sum\nolimits_{k = 1}^M {{{\left| {h_{SR}^k} \right|}^2}}
\end{equation}

Once the destination fails to decode the received signal $y_{SR}^k$, the destination will feed back a negative acknowledgement (NACK) to the source node. The next HARQ round will be immediately triggered and the source retransmits the same signal $x$ once receiving the NACK message. The number of HARQ-CC rounds is limited to $K$ to avoid network congestion. Only if the relay successfully decodes the received signal, the transmission procedure will proceed to the phase II, that is, the source and the relay are cooperative to transmit the signal. In this phase, the received signal associated with the link between the relay and the destination is
\begin{equation}\label{eqn:rec_signal_k_SR}
  y_{RD}^k = h_{RD}^k x + n_{RD}^k
\end{equation}
where $h_{RD}^k$ represents the channel coefficient associated with the link of relay-to-destination; $n_{RD}^k$ is defined as an additive white Gaussian noise with mean zero and variance $\sigma^2$, i.e. $n_{RD}^k \sim N(0,\sigma^2)$. Similarly, $|h_{RD}^{k}| \sim Nakagami\left( m,{\Omega _{RD}^k} \right)$, and the cross-correlation between $|h_{RD}^i|^2$ and $|h_{RD}^j|^2$  is denoted by ${\rho _{RD}^{i,j}} = {\lambda _{RD}^i}^2{\lambda _{RD}^j}^2$. Similar to the phase I, suppose that the relay successfully decodes the received signal after $r$ HARQ rounds, then the resultant combine SNR relied upon $\{y_{SD}^k, y_{RD}^l, k \in [1,M], l \in [r+1,M] \}$ is expressed as
\begin{equation}\label{eqn:snr_com_fin_d}
\gamma _{D,2}^{M,r} = \frac{{{P_s}}}{{{\sigma ^2}}}\sum\nolimits_{k = 1}^M {{{\left| {h_{SD}^k} \right|}^2}}  + \frac{{{P_s}}}{{{\sigma ^2}}}\sum\nolimits_{l = r + 1}^M {{{\left| {h_{SR}^l} \right|}^2}}
\end{equation}

It should be mentioned that $h_{SD}^k$, $h_{SR}^k$ and $h_{RD}^k$ are independent of each other. In the case of successively decoding the signal $x$ at the destination no matter in the phase I or the phase II, an acknowledgement (ACK) message will be transmitted to the source. Consequently, the retransmissions for $x$ will be ended, and the subsequent $b$ information bits will be encoded for transmission in the same way. To address the problem of rate selection, the distribution of the combine SNR should be derived first.

\section{Closed-Form PDF and CDF for Instantaneous Combining SNRs}
\label{sec:per}
From (\ref{eqn:snr_com_fina}), (\ref{eqn:snr_com_fin_sr}) and (\ref{eqn:snr_com_fin_d}), all the combine SNRs are expressed as the sum of correlated Gamma random variables (RVs). The nature of our work becomes to determine the distribution of sum of correlated Gamma random variables. Some previous literatures have studied this problem. In \cite{alouini2001sum,alexandropoulos2009new}, the PDF for the sum of correlated Gamma RVs is expressed as multi-fold infinite series. In \cite{chaianya2014adaptive}, the PDF for sum of correlated exponential (a special case of Gamma) RVs is expressed as a confluent Lauricella hypergeometric function of multivariates, which is multi-fold integration. However, both the results obtained in \cite{chaianya2014adaptive} and \cite{alouini2001sum} are hard to evaluate precisely because of involving of multiple folds. Imran S. A. in \cite{ansari2012sum} give a new result for the distribution of sum of independent Gamma RVs. Its PDF is derived as a generalized Fox's H function. It is noteworthy that this alternative result presented in \cite{ansari2012sum} can be expressed in the form of a single-fold integration, which is readily computable. The authors in \cite{ansari2012sum} provide an efficient MATHEMATICA implementation of the generalized Fox's $\bar H$ function, which can be evaluated fast and accurately. Following a similar approach as in \cite{ansari2012sum}, we here aim to derive the distribution of the sum of correlated Gamma RVs.

We first derive the distribution for $\gamma _{D,1}^M = \sum\nolimits_{k = 1}^M {{z_k}} $, where ${z_k} = \frac{{{P_s}}}{{{\sigma ^2}}}{\left| {{h_{SD}^k}} \right|^2}$. It can be proved that $z_k$ follows Gamma distribution with parameter $(m,{\Omega _{SD}^k}'/m)$, i.e. $z_k \sim \mathcal{G} (m,{\Omega_{SD}^k}'/{m})$, where ${\Omega_{SD}^k}'=\Omega_{SD}^k P_s/\sigma^2$. Without loss of generality, $z_k$ is distributed as $z_k \sim {\mathcal{G}}(m,\beta_k)$, where $\beta_k = {\Omega_{SD}^k}'/{m}$. After some manipulations, it can be obtained that the cross-correlation between $z_i$ and $z_j$ is equal to $\rho_{SD}^{i,j}$. To cope with this problem, the following lemma is given for the PDF of ${Y_M} = \sum\nolimits_{k = 1}^M {{z_k}} $.
\begin{lemma}\label{theorem:PDF}(PDF of the Sum of Gamma Random Variables with Exponential Correlation). The pdf of ${\gamma _{D,1}^M}$ can be derived in terms of the generalized Fox's $\bar H$ function \cite[Def. A.57]{mathai2009h} as
\begin{equation}\label{eqn:pdf_Y}
\begin{array}{l}
{f_{{\gamma _{D,1}^M}}^{(1)}}\left( y \right) = \prod\limits_{k = 1}^M {{\delta _k}^{ - m}}  \times \\
\bar H_{M,M}^{0,M}\left[ {\left. {\begin{array}{*{20}{c}}
{\left( {1 - {\delta _1}^{ - 1},1,m} \right), \cdots ,\left( {1 - {\delta _M}^{ - 1},1,m} \right)}\\
{\left( { - {\delta _1}^{ - 1},1,m} \right), \cdots ,\left( { - {\delta _M}^{ - 1},1,m} \right)}
\end{array}} \right|{e^y}} \right]
\end{array}
\end{equation}
where $\{\delta _k\}_{k=1}^{M}$ are defined as the eigenvalues of the matrix $\bf A=DC$, where $\bf D$ is the $M \times M$
diagonal matrix with the diagonal entries $\{\beta _k\}_{k=1}^{M}$, and $\bf C$ is the $M \times M$ positive definite matrix given by
\begin{equation}\label{eqn:C_mat_def}
{\bf C} = \left[ {\begin{array}{*{20}{c}}
1&{\sqrt {\rho _{SD}^{_{1,2}}} }& \cdots &{\sqrt {\rho _{SD}^{_{_{1,M}}}} }\\
{\sqrt {\rho _{SD}^{_{2,1}}} }&1& \cdots &{\sqrt {\rho _{SD}^{_{_{2,M}}}} }\\
 \vdots & \vdots & \ddots & \vdots \\
{\sqrt {\rho _{SD}^{_{M,1}}} }&{\sqrt {\rho _{SD}^{_{_{M,2}}}} }& \cdots &1
\end{array}} \right].
\end{equation}

\end{lemma}
\begin{proof}
Please refer to \cite{alouini2001sum} and \cite{ansari2012sum}.
\end{proof}

Accordingly, by using the result from Lemma \ref{theorem:PDF}, the cumulative distribution function (CDF) of ${Y_M}$ is given by the following Theorem.

\begin{theorem} \label{theorem:CDF}(CDF of the Sum of Gamma RVs with Exponential Correlation). For arbitrary $m$, the CDF of ${\gamma _{D,1}^M}$ can be derived in term of the Fox's $\bar{H}$ function as (\ref{eqn:CDF_Y}) as shown on the top of next page.
\begin{figure*}[!t]
\normalsize
\begin{equation}
\label{eqn:CDF_Y}
F_{\gamma _{D,1}^M}^{(1)}\left( y \right) = 1 - \prod\limits_{k = 1}^M {{\delta _k}^{ - m}} \bar H_{M + 1,M + 1}^{1,M}\left[ {\left. {\begin{array}{*{20}{c}}
{\overbrace {\left( {1 - {\delta _1}^{ - 1},1,m} \right), \cdots ,\left( {1 - {\delta _M}^{ - 1},1,m} \right)}^{M - bracketed\;terms},\left( {1,1} \right)}\\
{\left( {0,1} \right),\underbrace {\left( { - {\delta _1}^{ - 1},1,m} \right), \cdots ,\left( { - {\delta _M}^{ - 1},1,m} \right)}_{M - bracketed\;terms}}
\end{array}} \right|{e^y}} \right]
\end{equation}
\hrulefill
\vspace*{4pt}
\end{figure*}
\end{theorem}
\begin{proof}
See Appendix \ref{appen_CDF}.
\end{proof}

There are some special cases for the expression (\ref{eqn:CDF_Y}) as follows:
\begin{enumerate}
  \item In the case of $M=0$, ${F_{{\gamma _{D,1}^M}}^{(1)}}\left( y \right)=1$ holds since
  \begin{equation}\label{eqn:speci_M_0}
  F_{\gamma _{D,1}^M}^{(1)}\left( y \right) = 1 - \frac{1}{{2\pi i}}\oint_C {\frac{{{e^{ - sy}}}}{s}ds}  = 1 - u\left( { - y} \right)
  \end{equation}
  where $u\left( { y} \right)$ represents Heaviside step function.
  \item For $M=1$, because $\gamma _{D,1}^M = z_1$ follows Gamma distribution, it follows that
  \begin{equation}\label{eqn:speci_M_1}
F_{\gamma _{D,1}^M}^{(1)}\left( y \right) = \Pr \left( {{z_1} \le y} \right) = \frac{1}{{\Gamma \left( m \right)}}\Upsilon \left( {m,\frac{y}{{{\beta _1}}}} \right)
  \end{equation}
  where $\Upsilon \left( {k,\theta } \right)$ represents the lower imcomplete gamma distribution.
  \item If $m$ is an integer, the generalized Fox's $\bar H$ function reduces to Meijer-G function \cite[Def. 9.30]{gradshteyn1965table}, thus $F_{\gamma _{D,1}^M}^{(1)}\left( y \right)$ is rewritten in terms of Meijer-G function as (\ref{eqn:CDF_special_m_int}). Specifically, the corresponding CDF for Rayleigh fading channels can be obtained with $m=1$.
      \begin{figure*}[!t]
\normalsize
\begin{equation}
\label{eqn:CDF_special_m_int}
F_{\gamma _{D,1}^M}^{(1)}\left( y \right) = 1 - \prod\limits_{k = 1}^M {{\delta _k}^{ - m}} G_{mM + 1,mM + 1}^{mM,1}\left[ {\left. {\begin{array}{*{20}{c}}
{1,\overbrace {\overbrace {\left( {1 + {\delta _1}^{ - 1}} \right), \cdots \left( {1 + {\delta _1}^{ - 1}} \right)}^{m - times}, \cdots ,\overbrace {\left( {1 + {\delta _M}^{ - 1}} \right), \cdots ,\left( {1 + {\delta _M}^{ - 1}} \right)}^{m - times}}^{M - bracketed\;terms}}\\
{\underbrace {\underbrace {\left( {{\delta _1}^{ - 1}} \right), \cdots \left( {{\delta _1}^{ - 1}} \right)}_{m - times}, \cdots ,\underbrace {\left( {{\delta _M}^{ - 1}} \right), \cdots ,\left( {{\delta _M}^{ - 1}} \right)}_{m - times}}_{M - bracketed\;terms},0}
\end{array}} \right|{e^{ - y}}} \right]
\end{equation}
\hrulefill
\vspace*{4pt}
\end{figure*}
\end{enumerate}

In a similar way, the closed-form CDF $F_{\gamma _R^M}^{(1)}\left( y \right) $ for combine SNR ${\gamma _{R}^M}$ can be derived.
The remained work is to determine the distribution of $\gamma _{D,2}^{M,r}$. By using moment-generating function (MGF), the CDF of $\gamma _{D,2}^{M,r}$ is given in the following theorem. Without loss of generality, we define ${y_l} = \frac{{{P_s}}}{{{\sigma ^2}}}{\left| {h_{SR}^l} \right|^2}$, then $\gamma _{D,2}^{M,r} = \sum\nolimits_{k = 1}^M {{z_k}}  + \sum\nolimits_{l = r + 1}^M {{y_l}} $, where $z_k$ and $y_l$ are independent of each other, and ${\{y_l\}}_{l=1}^M$ are exponentially-correlated Gamma RVs. It follows that $y_l \sim \mathcal{G} (m,\theta_l)$, where $\theta_l  = \Omega _{RD}^l{P_s}/\left( {m{\sigma ^2}} \right)$, and the cross-correlation between $y_i$ and $y_j$ is equal to $\rho_{RD}^{i,j}$.

\begin{theorem} \label{theorem:CDF_mgf}(CDF of $\gamma _{D,2}^{M,r}$). For arbitrary $m$, the CDF of $\gamma _{D,2}^{M,r}$ can be expressed in term of the Fox's $\bar{H}$ function as (\ref{eqn:CDF_d_2}) on next page, where ${\{\alpha _k\}}_{k=r+1}^{M}$ are defined as the eigenvalues of the matrix $\bf B=FE$, where $\bf F$ is the $(M-r) \times (M-r)$
diagonal matrix with the diagonal entries $\{\theta _l\}_{l=r+1}^{M}$, and $\bf E$ is the $(M-r) \times (M-r)$ positive definite matrix given by
\begin{equation}\label{eqn:C_mat_def}
{\bf G} = \left[ {\begin{array}{*{20}{c}}
1&{\sqrt {\rho _{RD}^{_{r + 1,r + 2}}} }& \cdots &{\sqrt {\rho _{RD}^{_{_{r + 1,M}}}} }\\
{\sqrt {\rho _{RD}^{_{r + 2,r + 1}}} }&1& \cdots &{\sqrt {\rho _{RD}^{_{_{r + 2,M}}}} }\\
 \vdots & \vdots & \ddots & \vdots \\
{\sqrt {\rho _{RD}^{_{M,r + 1}}} }&{\sqrt {\rho _{RD}^{_{_{M,r + 2}}}} }& \cdots &1
\end{array}} \right]
\end{equation}
\begin{figure*}[!t]
\normalsize
\begin{equation}
\label{eqn:CDF_d_2}
\begin{array}{l}
{F_{\gamma _{D,2}^{M,r}}^{(2)}}\left( y \right) = 1 - \prod\limits_{k = 1}^M {{\delta _k}^{ - m}} \prod\limits_{l = r + 1}^M {{\alpha _k}^{ - m}}  \times \\
\bar H_{M + 1,M + 1}^{1,M}\left[ {\left. {\begin{array}{*{20}{c}}
{\overbrace {\left( {1 - {\delta _1}^{ - 1},1,m} \right), \cdots ,\left( {1 - {\delta _M}^{ - 1},1,m} \right)}^{M - bracketed\;terms},\overbrace {\left( {1 - {\alpha _{r + 1}}^{ - 1},1,m} \right), \cdots ,\left( {1 - {\alpha _M}^{ - 1},1,m} \right)}^{\left( {M - r} \right) - bracketed\;terms},\left( {1,1} \right)}\\
{\left( {0,1} \right),\underbrace {\left( { - {\delta _1}^{ - 1},1,m} \right), \cdots ,\left( { - {\delta _M}^{ - 1},1,m} \right)}_{M - bracketed\;terms},\underbrace {\left( { - {\alpha _{r + 1}}^{ - 1},1,m} \right), \cdots ,\left( { - {\alpha _M}^{ - 1},1,m} \right)}_{\left( {M - r} \right) - bracketed\;terms}}
\end{array}} \right|{e^y}} \right]
\end{array}
\end{equation}
\hrulefill
\vspace*{4pt}
\end{figure*}
\end{theorem}
\begin{proof}
See Appendix \ref{appen_CDF_mgf}.
\end{proof}

\section{Optimal Rate Selection}
\label{adap_sch}
An optimal rate selection scheme is proposed for cooperative HARQ-CC system over time-correlated fading channels via maximizing the DLT. Given the maximum allowable number of transmissions for a single packet $K$ and data transmission rate $\mathcal R$, the DLT is given by \cite{kim2008optimal}
\begin{equation}\label{eqn:dlt_def}
{\mathcal{T}_K} = \sum\limits_{k = 1}^K {\frac{\mathcal R}{k}} \left( {{\rm{P}}_D^{out}(k - 1) - {\rm{P}}_D^{out}(k)} \right)
\end{equation}
where ${{\rm{P}}_D^{out}(k)}$ defines the outage probability that the destination fails to decode the received signals after $k$ HARQ round. Specifically, ${{\rm{P}}_D^{out}(0)} = 1$, since it is impossible to decode the packet at round $0$. By using the law of total probability, ${{\rm{P}}_D^{out}(k)}$ is formulated as
\begin{equation}\label{eqn:out_total}
{\rm{P}}_D^{out}(k) = \sum\limits_{r = 1}^\infty  {\left( {{\rm{P}}_R^{out}(r - 1) - {\rm{P}}_R^{out}(r)} \right){\rm{P}}_D^{out}(k|r)} , 1 < k \le K
\end{equation}
where ${{\rm{P}}_R^{out}(r)}$ defines the outage probability at the relay after $r$ HARQ rounds, and ${{\rm{P}}_D^{out}(k|r)}$ represents the outage probability at the destination after $k$ HARQ rounds given that the relay successfully decodes the packet at round $r$. Therefore, to address the rate selection, the outage probability ${{\rm{P}}_R^{out}(r)}$ and ${{\rm{P}}_D^{out}(k|r)}$ should be derived first.
\subsection{Outage Probability}
For ${{\rm{P}}_R^{out}(r)}$, it is written as
\begin{equation}\label{eqn:out_pro}
{P_R^{out}\left( r \right) = \Pr \left( {I_R^r < \mathcal{R}} \right)}
\end{equation}
where $I_R^M$ represents the total accumulated mutual information for MRC at the relay till $M$ HARQ rounds, which is evaluated as
\begin{equation}\label{eqn:accumu_mul_inf}
I_R^r = {\log _2}\left( {1 + \gamma _R^r} \right)
\end{equation}

By substituting (\ref{eqn:accumu_mul_inf}) into (\ref{eqn:out_pro}), it yields
\begin{equation}\label{eqn:outage_final_rel}
P_R^{out}\left( r \right) = \Pr \left( {\gamma _R^r < {2^\mathcal{R}} - 1} \right) = F_{\gamma _R^r}^{(1)}\left( {{2^\mathcal{R}} - 1} \right)
\end{equation}
The last step holds by applying the Theorem \ref{theorem:CDF}.

For ${{\rm{P}}_D^{out}(k|r)}$, it is given by
\begin{equation}\label{eqn:out_de_def}
{\rm{P}}_D^{out}(k|r) = \Pr \left( {I_D^{k,r} < {\cal R}} \right)
\end{equation}
where $I_D^{k,r}$ represents the total accumulated mutual information for MRC at the destination till $M$ HARQ rounds given that the relay successfully decodes the packet at the $r$th round, and is formulated as
\begin{equation}\label{eqn_mul_des_def}
I_D^{k,r} = \left\{ \begin{array}{l}
{\log _2}\left( {1 + \gamma _{D,1}^k} \right),r \ge k\\
{\log _2}\left( {1 + \gamma _{D,2}^{k,r}} \right),else
\end{array} \right.
\end{equation}

By putting (\ref{eqn_mul_des_def}) into (\ref{eqn:out_de_def}) and then applying the Theorems \ref{theorem:CDF} and \ref{theorem:CDF_mgf}, it follows
\begin{equation}\label{eqn:out_des_final}
{\rm{P}}_D^{out}(k|r) = \left\{ \begin{array}{l}
F_{\gamma _{D,1}^k}^{(1)}\left( {{2^{\cal R}} - 1} \right),r \ge k\\
F_{\gamma _{D,2}^{k,r}}^{(2)}\left( {{2^{\cal R}} - 1} \right),else
\end{array} \right.
\end{equation}

Hereby, by substituting (\ref{eqn:out_de_def}) and (\ref{eqn:out_des_final}) into (\ref{eqn:out_total}), ${{\rm{P}}_D^{out}(k)}$ can be expressed as
(\ref{eqn:CDF_out_des_fin}) on next page.
\begin{figure*}[!t]
\normalsize
\begin{equation}
\label{eqn:CDF_out_des_fin}
{\rm{P}}_D^{out}(k) = \sum\limits_{r = 1}^{k - 1} {\left( {F_{\gamma _R^{r - 1}}^{(1)}\left( {{2^{\cal R}} - 1} \right) - F_{\gamma _R^r}^{(1)}\left( {{2^{\cal R}} - 1} \right)} \right)F_{\gamma _{D,2}^{k,r}}^{(2)}\left( {{2^{\cal R}} - 1} \right)}  + F_{\gamma _{D,1}^k}^{(1)}\left( {{2^{\cal R}} - 1} \right)F_{\gamma _R^{k - 1}}^{(1)}\left( {{2^{\cal R}} - 1} \right)
\end{equation}
\hrulefill
\vspace*{4pt}
\end{figure*}

In order to study the impact of channel correlation and fading severity on cooperative HARQ-CC system, we consider a common scenario for numerical analysis. Assume $\gamma_T = P_s/\sigma^2$, which is termed as transmit SNR. The mean power of each channel coefficient is assumed to be $\mathbb{E}(|h_{SD}^k|)=0.5$, $\mathbb{E}(|h_{SR}^k|)=\mathbb{E}(|h_{RD}^k|)=1$. Under these assumptions, the outage probability $P_{D}^{out}(K)$ is plotted versus transmit SNR $\gamma_T$ for different $\rho$, $m$ and $K$ by setting $R=2~\rm{bps/Hz}$ in Fig. \ref{fig:outage probability}. The presented results show a perfect match between the analytical and simulation results for $\rho=0.2,~0.8$ and $m=1,~2$. From this figure, it can be observed that the outage probability decreases as transmit SNR $\gamma_T$, maximum transmissions $K$ and fading-order parameter $m$ increase. Moreover, the outage probability increases as the correlation coefficient $\rho$ increases.
\begin{figure}
  \centering
  \includegraphics[width=3in]{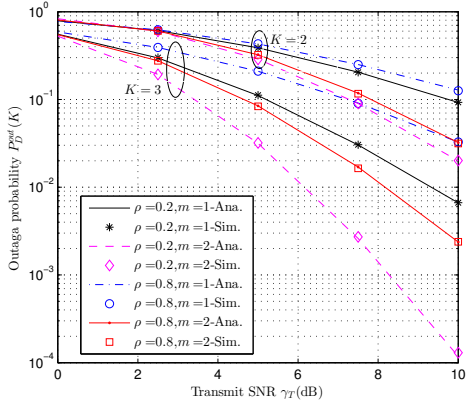}\\
  \caption{Outage probability $P_D^{out}(K)$ vs. transmit SNR $\gamma_T$.}\label{fig:outage probability}
\end{figure}

\subsection{Optimal Rate}
To find the optimal rate, the relation between DLT $\mathcal{T}$ and transmission rate $\mathcal{R}$ is investigated first. An example is given in Fig. \ref{fig:dlt_r} by setting $m=2$ and $\gamma_T=0~\rm{dB}$. From this figure, it can be observed that each curve has only one maximum point. This sole maximum point can be effectively obtained by adopting Golden section search.
\begin{figure}
  \centering
  \includegraphics[width=3in]{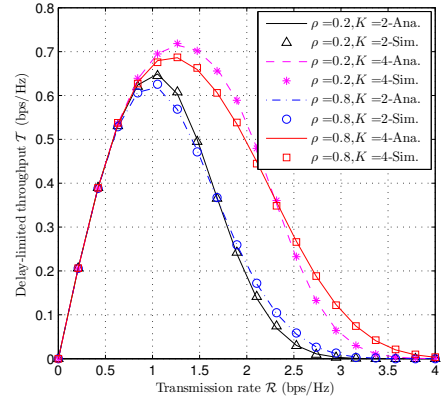}\\
  \caption{DLT $\mathcal{T}$ vs. transmission rate $\mathcal{R}.$}\label{fig:dlt_r}
\end{figure}
Hereby, for transmitting certain $b$ bits information, the DLT is maximized by selecting an optimal rate which is determined from the following optimization problem:
\begin{equation*}
\begin{aligned}
& \underset{\mathcal{R}}{\text{maximize}}
& & \mathcal{T}_M \\
& \text{subject to}
& & \mathcal{R} \ge 0.
\end{aligned}
\end{equation*}

As exhibited in Fig. \ref{fig:dlt} for a fixed value of $m=2$, the optimal rate for different $K$ and $\rho$ is plotted against transmit SNR. It can shown that the optimal rate increases as $\gamma_T$ increases. For instance, the optimal rate increases from $1~\rm{bps/Hz}$ to $2.2~\rm{bps/Hz}$ when $\gamma_T$ varies from $0$ to $10~dB$ for $\rho=0.2$ and $K=2$. If the maximum transmissions $K$ is increased, the optimal rate increases as well. Moreover, the optimal rate decreases with the increase of time-correlation of channels. It is due to the fact that the channels with high time-correlation will degrade the time diversity gain of HARQ.
\begin{figure}
  \centering
  \includegraphics[width=3in]{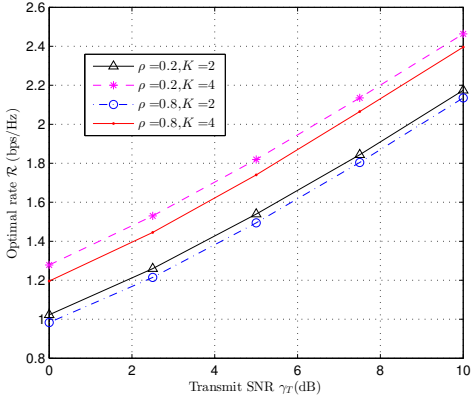}\\
  \caption{Optimal transmission rate $\mathcal{R}$ vs. transmit SNR $\gamma_T$.}\label{fig:dlt}
\end{figure}
\section{Conclusions}
\label{sec_con}
In this paper, an optimal rate selection scheme has been proposed for cooperative HARQ-CC systems, where the general time-correlated Nakagami-m fading model is considered. The involvement of cooperative communication and the presence of time-correlation make the problem challenging. The nature of rate selection is to choose the appropriate transmission rate to maximize the DLT. To cope with this problem, it is necessary to determine the distribution of the combine SNRs via MRC. By means of the generalized Fox's $\bar H$ function, the PDF and CDF of combine SNRs have been derived, with which the outage probability and DLT are then obtained.

These analytical results have been validated by Monte Carlo simulations. It has been found from the results that the system can achieve more diversity gain from the less correlated channels, and the outage probability decreases with the increase of fading-order parameter $m$, and etc. It is noteworthy that the metrics for evaluating the performance of the cooperative HARQ-CC system are derived in terms of the generalized Fox's H function, which has a low computational complexity and can be evaluated accurately. Finally, the optimal rate to maximize the DLT has been presented.
\appendices
\section{Proof of Theorem \ref{theorem:CDF}}
\label{appen_CDF}
By integrating the PDF expressed in (\ref{eqn:pdf_Y}) from $0$ through $y$, the CDF of ${\gamma _{D,1}^M}$ is obtained, more precisely,
\begin{equation}\label{eqn:CDF_Y_der}
\begin{array}{l}
{F_{{{\gamma _{D,1}^M}}}^{(1)}}\left( y \right) = 1 - \prod\limits_{k = 1}^M {{\delta _k}^{ - m}} \frac{1}{{2\pi i}}\int_y^\infty  {\oint_{\cal C} {\frac{{\prod\limits_{k = 1}^M {{\Gamma ^m}\left( {{\delta _k}^{ - 1} - s} \right)} }}{{\prod\limits_{k = 1}^M {{\Gamma ^m}\left( {1 + {\delta _k}^{ - 1} - s} \right)} }}{e^{ - st}}ds} dt} \\
{\rm{ = }}1 - \prod\limits_{k = 1}^M {{\delta _k}^{ - m}} \frac{1}{{2\pi i}}\oint_{\cal C} {\frac{{\Gamma \left( s \right)\prod\limits_{k = 1}^M {{\Gamma ^m}\left( {{\delta _k}^{ - 1} - s} \right)} }}{{\Gamma \left( {1 + s} \right)\prod\limits_{k = 1}^M {{\Gamma ^m}\left( {1 + {\delta _k}^{ - 1} - s} \right)} }}{e^{ - sy}}ds} \\
 = 1 - \prod\limits_{k = 1}^M {{\delta _k}^{ - m}} \frac{1}{{2\pi i}}\oint_{\cal C} {\frac{{\Gamma \left( { - s} \right)\prod\limits_{k = 1}^M {{\Gamma ^m}\left( {{\delta _k}^{ - 1} + s} \right)} }}{{\Gamma \left( {1 - s} \right)\prod\limits_{k = 1}^M {{\Gamma ^m}\left( {1 + {\delta _k}^{ - 1} + s} \right)} }}{e^{sy}}ds}
\end{array}
\end{equation}
After identifying the Mellin-Barnes contour integral in the above equation with the generalized Fox's $\bar{H}$ function, the CDF of ${\gamma _{D,1}^M}$ is obtained as (\ref{eqn:CDF_Y}).
\section{Proof of Theorem \ref{theorem:CDF_mgf}}
\label{appen_CDF_mgf}
Due to the independence of $\sum\nolimits_{k = 1}^M {{z_k}} $ and $\sum\nolimits_{l = r + 1}^M {{y_l}} $, the MGF of $\gamma _{D,2}^{M,r}$ is given by
\begin{equation}\label{eqn:mgf_d_2}
\begin{array}{l}
{M_{\gamma _{D,2}^{M,r}}}(s) = \prod\nolimits_{k = 1}^M {{\delta _k}^{ - m}{{\left( {{\delta _k}^{ - 1} + s} \right)}^{ - m}}}  \times \\
\prod\nolimits_{l = r + 1}^M {{\alpha _l}^{ - m}{{\left( {{\alpha _l}^{ - 1} + s} \right)}^{ - m}}}
\end{array}
\end{equation}

The PDF of $\gamma _{D,2}^{M,r}$ can be straightforward derived as a generalized Fox's $\bar H$ function via inverse Laplace transform. Furthermore, the CDF of $\gamma _{D,2}^{M,r}$ is derived as (\ref{eqn:CDF_d_2}) by using the same method from Theorem \ref{theorem:CDF}.

\section*{Acknowledgement}
This work was supported jointly by the Macao Science and Technology Development Fund under grant 067/2013/A and the Research Committee of University of Macau under grants MRG023, MYRG078 and MYRG101.

\bibliographystyle{ieeetr}
\bibliography{HARQ-CC}

\end{document}